\documentclass[12pt]{article}

\usepackage[titletoc,title]{appendix}
\usepackage[latin1]{inputenc}
\usepackage{amsmath}
\usepackage{amssymb}
\usepackage{booktabs}
\usepackage{graphicx}
\usepackage[margin=1.25in]{geometry}
\usepackage[bottom]{footmisc}
\usepackage{indentfirst}
\usepackage{endnotes}
\usepackage{mathabx}
\usepackage{enumerate}
\usepackage{rotating}
\usepackage{multirow}
\usepackage{booktabs,calc}
\usepackage{threeparttable}
\usepackage{float}
\usepackage{url}
\usepackage{amsmath,amsfonts}
\usepackage{amsthm}
\usepackage[onehalfspacing]{setspace}

\newcommand{\suchthat}{\;\ifnum\currentgrouptype=16 \middle\fi|\;}

\newtheorem{assm}{Assumption}
\newtheorem{theorem}{Theorem}
\newtheorem{prop}{Proposition}
\newtheorem{lemma}{Lemma}
\newtheorem{cor}{Corollary}

\newtheorem{example}{Example}

\theoremstyle{remark}
\newtheorem{remark}{Remark}

\usepackage{chngcntr}
\usepackage{graphicx}
\usepackage{amsmath}
\usepackage{thmtools}
\usepackage{amssymb}
\usepackage{parskip}
\usepackage{sgame}
\usepackage{color}
\usepackage{tikz}
\usetikzlibrary{trees,calc}

\usepackage{xcolor}

\DeclareMathOperator*{\argmax}{\arg\!\max}

\newcommand{\E}{\mathbb{E}}

\newcommand{\eps}{\varepsilon}

\newcommand{\oD}{\overline{D}}

\newcommand{\rB}{\overline{B}}
\usepackage{natbib}

\usepackage{url,color,hyperref}
\definecolor{darkblue}{rgb}{0.0,0.0,0.35}
\hypersetup{colorlinks,breaklinks,
            linkcolor=darkblue,urlcolor=darkblue,
            anchorcolor=darkblue,citecolor=darkblue}

\begin{document}

\title{\large{Identification of Random Coefficient Latent Utility Models}\thanks{Any remaining errors are our own.} }

\author{ \small{Roy Allen}  \\
    \small{Department of Economics} \\
    \small{University of Western Ontario} \\
    \small{rallen46@uwo.ca}
    \and 
    \small{John Rehbeck} \\
    \small{Department of Economics} \\
    \small{The Ohio State University} \\
    \small{rehbeck.7@osu.edu}
}
\date{\small{ \today }} 

\maketitle

\begin{abstract}
This paper provides nonparametric identification results for random coefficient distributions in perturbed utility models. We cover discrete and continuous choice models. We establish identification using variation in mean quantities, and the results apply when an analyst observes aggregate demands but not whether goods are chosen together. We require exclusion restrictions and independence between random slope coefficients and random intercepts. We do not require regressors to have large supports or parametric assumptions.
\end{abstract}

\newpage

\section{Introduction}

Latent utility models with linear random coefficients have been extensively used. They have a long history in discrete choice,\footnote{\cite{heckman2001micro} attributes the first use to \cite{domenich1975urban} in economics.} and have become increasingly popular due to computational advances (see e.g. \cite{train2009discrete}). They now form the core demand system of most applied work involving demand for differentiated products following \cite{berry1995automobile}. Progress has been made on identification of these models in discrete choice, but gaps remain, even in semiparametric settings. For example, nonparametric identification of the distribution of random coefficients in the random coefficients nested logit model has not been established without unbounded regressors.\footnote{See \cite{nevo2000practitioner}, p. 524-526. \cite{il2014identification} has established identification in the special case where intercept location coefficients are $0$.} More broadly, there is growing interest in models that allow complementarity, but even less is known about identification of random coefficients in these models.\footnote{Recent work includes \cite{gentzkow2007}, \cite{mcfadden2012theory}, \cite{fosgerau2019inverse}, \cite{allen2019identification}, \cite{ershov2018mergers}, \cite{monardo2019flexible}, \cite{iaria2019}, and \cite{wang2020}. This work is an outgrowth of the discrete choice additive random utility model \citep{mcfadden1981} and differs from classic continuous demand systems (e.g. \cite{deaton1980almost}) by focusing on characteristic variation rather than variation in a budget constraint.}

The main contribution of this paper establishes nonparametric identification for the moments of random coefficients in a general class of latent utility models. The framework applies to discrete and continuous choices. As a special case, we establish identification for a bundles model with limited consideration of either alternatives or characteristics (Example~\ref{ex:bundles}). Identification only depends on the average structural function \citep{blundell2003endogeneity}. Thus the results can be applied when one observes the average demands of individuals without observing whether goods are chosen together. Leveraging the main result, we can identify the distribution of random coefficients when it is characterized by its moments (e.g. normal distributions).  Specialized to discrete choice, the main contribution is new since it does not require any regressor to be unbounded.

Two key ingredients let us get traction for identification. First, we assume independence between random slopes and random intercepts. This is a standard assumption in the widely-used random coefficients logit model. We use this assumption to integrate out the random intercepts. This smooths out demand when conditioning on non-intercept components of the utility function (regressors and random coefficients). This also allows us to treat discrete and continuous choice models in a common framework. Second, we exploit theoretical restrictions since choices arise from optimization. Without using this structure, the model would resemble general random coefficients index models studied in \cite{fox2012random} and \cite{lewbel2017unobserved}. While they assume a function governing the mapping from indices to choices is \textit{known}, we do not.\footnote{In discrete choice, assuming this function is known translates to the distribution of random intercepts being known (e.g. logit). \cite{lewbel2017unobserved} show that one can drop the assumption that this mapping is known in some settings if one imposes an additional additive separability assumption.}

The fundamental shape restriction we exploit is that, after integrating out random intercepts, integrated mean choices are the derivative of a convex function. This follows from an application of the envelope theorem. Similar tools from convex analysis have also been used for identification of hedonic models \citep{ekeland2002identifying,ekeland2004identification,heckman2010nonparametric,chernozhukov2019single}, matching \citep{galichon2015cupid}, dynamic discrete choice \citep{chiong2016duality}, discrete choice panel models \citep{shi2018estimating}, and perturbed models with additively separable heterogeneity \citep{allen2019identification}, among others.\footnote{\cite{matzkin1994restrictions} reviews other identification results using shape restrictions motivated by economic theory. See also work on optimal transport, as in \cite{galichon2018optimal}.}

By exploiting the envelope theorem, we can treat several models in a common framework. There is little work on identification of optimizing models with linear random coefficients outside of discrete choice. Exceptions include \cite{dunker2018nonparametric} for discrete games and \cite{dunker2017nonparametric} and \cite{iaria2019} for a random coefficients version of the \cite{gentzkow2007} discrete bundles model. We differ by requiring identification of only the average structural function (``mean demands''), without needing to observe the frequency with which goods are chosen together.\footnote{\cite{wang2020} also works with the average structural function but does not identify the distribution of random coefficients.} Identification with linear random coefficients has also been established in settings without assuming an optimizing model. See for example the simultaneous equations analysis in  \cite{masten2017random} and references therein.

Identification of linear random coefficients has been extensively studied in discrete choice. Despite this, nonparametric identification has only been established either requiring a regressor with large support or assuming the distribution of random intercepts is either known or parametric. One reason we do not require a large support assumption is that we focus on identification of the distribution of random coefficients without identifying random intercepts. Papers that make use of large support regressors with heterogeneity that is not additively separable include \citep{ichimura1998maximum}, \cite{berry2009nonparametric}, \cite{briesch2010nonparametric}, \cite{gautier2013nonparametric}, \cite{fox2016nonparametric}, \cite{dunker2017nonparametric}, and \cite{fox2017note}.\footnote{Exceptions include \cite{kashaev2018identification} and \cite{matzkin2019constructive}, but neither paper studies nonparametric identification for the distribution of linear random coefficients.} Several of these papers additionally assume large support regressors \textit{also} have the same coefficient across goods. In contrast, \cite{fox2012random} and \cite{il2014identification} do not assume large support or  a homogeneous regressor, but assume the distribution of the random intercept is known (e.g. logit). \cite{chernozhukov2019nonseparable} discusses identification of ratios of certain moments of the distribution of random coefficients without requiring large support, but do not provide conditions under which the full distribution is identified.

The remainder of the paper proceeds as follows. Section~\ref{sec:setup} provides details on the class of latent utility models we study and examples of behavior that this covers. Section~\ref{sec:main} provides the main result, which identifies arbitrary order moments of random coefficients and shows that a single independence and scale assumption can be used to identify all other moments. Section~\ref{sec:welfare} discusses how to recover different welfare objects and perform counterfactuals. Finally, Section~\ref{sec:discussion} discusses relations to some existing papers, shows how the results can be taken to settings with non-linear random coefficients, and discusses some testable properties of the framework.  

\section{Setup}\label{sec:setup}

This paper studies the \textit{random coefficients perturbed utility model}, in which optimizing choices satisfy
\begin{equation} \label{eq:latent}
Y(X,\beta,\eps) \in \argmax_{y \in B} \sum_{k = 1}^K y_k (\beta'_k X_k) + D(y, \eps).
\end{equation}
\cite{mcfadden2012theory} and \cite{allen2019identification} have studied related frameworks without random coefficients. We interpret $Y(\cdot)$ as the quantity vector for $K$ different goods. The vector $X_k = (X_{k,1}, \ldots, X_{k,d_k})'$ denotes observable shifters of the desirability of good $k$, and $\beta_k = (\beta_{k,1}, \ldots, \beta_{k,d_k})'$ denotes random coefficients on these shifters, which may be good-specific. The index $\beta'_k X_k$ shifts the marginal utility of good $k$. We collect $X = (X'_1, \ldots, X'_K)'$ and $\beta = (\beta'_1, \ldots, \beta'_K)'$. The term $D(y,\eps)$ is a disturbance that depends on unobservables $\eps$ of unrestricted dimension. When $D(y,\eps)=\sum_{k=1}^K \eps_k y_k$, $\eps_k$ can be interpreted as a random intercept for the desirability of the $k$-th good. In general, we refer to $D(y,\eps)$ as the random intercept. The set $B \subseteq \mathbb{R}^K$ is a feasibility set. This is introduced purely for exposition, since $D(y,\eps)$ can be $-\infty$ which allows random feasibility sets.

The focus of this paper is on identification of moments of the distribution of $\beta$. Our results do not require specification of the budget $B$, the disturbance $D$, or the distribution of over $\eps$. For concreteness, we provide some examples.
\begin{example}[Discrete Choice] \label{ex:disc}
Consider a discrete choice models with latent utility for good $k$ of the form
\[
v_k = \beta'_k X_k + \eps_k.
\]
When $\beta$ is random, this is a linear random coefficients model as studied in \cite{hausman1978conditional}, \cite{boyd1980effect}, \cite{cardell1980measuring}, among many others. This fits into the setup of (\ref{eq:latent}) by setting $D(y,\eps) = \sum_{k = 1}^K y_k \eps_k$, $B = \{ y \in \mathbb{R}^K \mid \sum_{k = 1}^K y_k = 1, y_k \geq 0 \}$ the probability simplex,\footnote{This allows the agent to randomize when there are utility ties.} and letting $Y(X,\beta,\eps) \in \{0,1 \}^K$ be a vector of indicators denoting which good is chosen. In many applications, an ``outside good'' is set to have a utility of $0$. This can be mapped to our setup by replacing the budget with $B = \{ y \in \mathbb{R}^K \mid \sum_{k = 1}^K y_k \leq 1, y_k \geq 0 \}$; this allows $Y(X,\beta,\eps) = (0,\ldots,0) \in B$, which can be interpreted as the choice of the outside option.

We also cover what is sometimes called the perturbed representation of choice, which can model market shares or individuals who like variety. For example, \cite{anderson1988representative} show logit models are related to the maximization problem
\[
\max_{y \in \Delta} \sum_{k = 1}^K y_k (\beta'_k X_k) + \sum_{k = 1}^K y_k \log(y_k),
\]
with $\Delta$ the probability simplex. \cite{hofbauer2002global} show by replacing the additive entropy term with a general disturbance, the setup covers all discrete choice additive random models once random intercepts are integrated out. \cite{fudenberg2015stochastic} study a model in which the disturbance is additively separable. \cite{fosgerau2019inverse} and \cite{allen2019revealed} show how to model complementarity with the perturbed utility representation.\\
\end{example}

\begin{example}[Bundles with Limited Consideration] \label{ex:bundles}
\cite{gentzkow2007} presents a model of choice of bundles involving online and print news. The model involves multiple goods and individuals can choose more than one good at the same time. A random coefficients version of the model has been studied in \cite{dunker2017nonparametric}. Let $v_{j,k}$ denote the utility associated with quantity $j$ of the first good, and quantity $k$ of the second good. Specify utilities
\begin{align*}
    \begin{split}
        v_{0,0} & = 0 \\
        v_{1,0} & = \beta'_1 X_1 + \eps_{1,0} \\
        v_{0,1} & = \beta'_2 X_2 + \eps_{0,1} \\
        v_{1,1} & = v_{1,0} + v_{0,1} + \eps_{1,1}.
    \end{split}
\end{align*}
Here, $\eps_{1,1}$ denotes a utility boost or loss from purchasing both goods relative to the sum of their individual utility. It describes complementarity/substitutability between the goods. For each quantity vector $\vec{y} = (y_1, y_2)$, set the utility as
\[
\sum_{k = 1}^2 y_k (\beta'_k X_k) + \left( y_1 \eps_{1,0} + y_2 \eps_{0,1} + y_1 y_2 \eps_{1,1} \right),
\]
and let the budget be $B = \{0,1\}^2$. Then the optimizing quantity vector $Y(X,\beta,\eps) \in \{ 0 , 1 \}^2$ fits into the setup of (\ref{eq:latent}).

This can be modified to include latent budgets. One may interpret these as mental ``consideration sets'' \citep{eliaz2011consideration, masatlioglu2012revealed, manzini2014stochastic, aguiar2017random} or general latent feasibility sets \citep{manski1977structure,conlon2013demand,brady2016menu}. In addition, we can allow limited consideration of the characteristics of goods.\footnote{See \cite{GABAIX2019261} for a survey of ``behavioral inattention.''}

To model these types of limited attention, consider a version of the bundles model  given by
\[
Y(X,\beta,\eps) \in \argmax_{y \in \{0,1\}^2} \sum_{k = 1}^2 y_k (\beta'_k X_k) + D(y,\eps),
\]
where
\[
D(y,\eps) = \begin{cases}  y_1 \eps_{1,0} + y_2 \eps_{0,1} + y_1 y_2 \eps_{1,1} & \text{if } y \in B(\eps) \\
    -\infty & \text{otherwise}
    \end{cases}.
\]
Here, the set $B(\eps)$ is a latent feasibility set, which could arise because an individual may not consider all goods or the analyst cannot observe when goods are out of stock. Some components of $\beta$ can be zero with positive probability, reflecting that individuals may not notice or care about certain characteristics.

This setup can be generalized to allow more goods, with some goods continuous and some goods discrete. What is key for our analysis is that the index $\beta'_k X_k$ shifts (only) the marginal utility of good $k$.
\end{example}

\subsection{Average Structural Function and Endogeneity}

This paper establishes identification of moments of $\beta$ using the average structural function \citep{blundell2003endogeneity}
\[
\overline{Y}(x) = \int Y(x, \beta, \eps) d \tau(\beta,\eps)
\]
for some probability measure $\tau$ that does not depend on covariates $x$. We assume that the measure $\tau$ satisfies a key independence condition.
\begin{assm}[Slope-Intercept Independence] \label{a:factor}
The random variables $\beta$ and $\eps$ are independent under the measure $\tau$, and the average structural function is finite.
\end{assm}
While independence between $\beta$ and $\eps$ is restrictive, it is a standard assumption in applications of the random coefficients logit model in discrete choice. It has been exploited for identification in \citep{fox2012random} and \citep{chernozhukov2019nonseparable}.\footnote{However, independence is not imposed in some papers studying identification. For example, \cite{ichimura1998maximum} or \cite{gautier2013nonparametric} do not impose independence of the slope and intercept.}

With this assumption, we can write
\[
\overline{Y}(x) = \int \int Y(x, \beta, \eps) d\mu(\eps) d\nu (\beta)
\]
for some probability measures $\mu$ and $\nu$. Technically, full independence is not needed as long as we can factor the average structural function in this way.

For an example of an average structural function, suppose $(Y,X,\beta,\eps)$ are random variables that satisfy $Y = Y(X, \beta, \eps)$ almost surely. Moreover, assume $X$, $\beta$, and $\eps$ are all independent.  In addition to independence, suppose a continuous version of the conditional mean of $Y$ given $X$ exists. Then
\[
\E[ Y \mid X = x] = \overline{Y}(x) = \int \int Y(x, \beta, \eps) d\mu(\eps) d\nu (\beta)
\]
for $x$ in the support of $X$,\footnote{Recall that the support of $X$ is the smallest closed set $S$ such that $P(X \in S) = 1$.} where $\mu$ is the marginal distribution of $\eps$ and $\nu$ is the marginal distribution of $\beta$.

The results in this paper apply to general average structural functions $\overline{Y}(x)$, not only the conditional mean. Thus, while slope-intercept independence is important for our results, independence between $X$ and $(\beta,\eps)$ is not. Therefore, the results in this paper are relevant for settings with endogeneity.

The goal of this paper is not to provide a new method to identify the average structural function, but rather to use the function to identify other features of a utility maximizing model. There is a large literature on identifying structural functions. \cite{blundell2003endogeneity} describe how to use control functions to identify the average structural function $\overline{Y}(x)$. \cite{altonji2005cross} identify derivatives of the average structural function using certain conditional independence or symmetry conditions. \cite{berry1994estimating}, \cite{berry1995automobile}, \cite{newey2003instrumental}, \cite{berry2014identification}, and \cite{dunker2017nonparametric} among others use instrumental variables to identify an average structural function from aggregate data.\footnote{A key step to apply these methods is injectivity in a market-level observable to a vector of unobservable endogenous vectors, usually denoted $\xi$. See \cite{allen2019injectivity} or Lemma 3 in \cite{allen2019identification} for injectivity results that cover the present model when the utility index for good $k$ is  $\beta'_k x_k + \xi_k$. Related injectivity results have appeared in \cite{galichon2015cupid} and \cite{chiong2017counterfactual}.} 

An important feature of the analysis is that only the average structural function is required to be identified over an appropriate region. Thus, the full distribution of $Y(x,\cdot,\cdot)$ induced by the product measure $\mu \times \nu$ over $(\beta, \eps)$ is not necessary for identification. For common discrete choice models the average structural function and the full distribution of $Y(x, \cdot,\cdot)$ contain the same information, but this is not true in general. This is particularly important when combining this analysis with work allowing endogeneity between $X$ and $(\beta, \eps)$. In particular, there are well-understood methods to identify the average structural function in the presence of endogeneity as mentioned earlier. In contrast, less is known about identification of the entire distribution of $Y(x,\cdot,\cdot)$ in the presence of endogeneity.\footnote{\cite{imbens2009identification} identify average and quantile structural functions with multidimensional heterogeneity in the outcome equation. \cite{torgovitsky2015identification} and \cite{d2015identification} identify the entire structural function with one-dimensional unobservable heterogeneity in the outcome equation. A multidimensional counterpart has been studied in \cite{fguns2019}. These papers all identify features of structural functions in the presence of endogeneity.}

In addition, requiring only the average structural function implies that the analysis can be applied to settings outside of discrete choice \textit{without} observing whether goods are chosen together. Of course if the full distribution of $Y(x,\cdot,\cdot)$ is identified, then these results apply as well. We recall that this paper does not study identification of the distribution of $\eps$ in the original latent utility model (\ref{eq:latent}). However, it is possible to identify the distribution of $\eps$ in some cases. For example, \cite{dunker2017nonparametric} show how to identify the distribution of random intercepts in a full-consideration random coefficients bundles model, provided the analyst has aggregate data on the frequency with which goods are chosen together.  

\subsection{Technical Tools}
We make use of an aggregation result that first integrates out the distribution of $\eps$.
\begin{lemma}[\cite{allen2019identification}] \label{l:agg} Let $Y(\cdot)$ satisfy (\ref{eq:latent}). For any measure $\mu$ over $\eps$ such that $\int Y(x, \beta, \eps) d\mu (\eps)$ and $\int D( Y(x, \beta, \eps), \eps) d\mu (\eps)$ exist and are finite, it follows that
\[
\int Y(x, \beta, \eps) \mu (d \eps) \in \argmax_{y \in \overline{B}} \sum_{k = 1}^K y_k (\beta'_k x_k) + \overline{D}(y),
\]
for $\rB$ the convex hull of $B$, and $\oD(y) = \sup_{\tilde{Y} \in \mathcal{Y} : \int \tilde{Y}(\eps) d\mu (\eps) = y} \int D \left(\tilde{Y}(\eps), \eps \right) d\mu ( \eps)$,\footnote{The supremum is taken to be $-\infty$ when there is no $\tilde{Y} \in \mathcal{Y}$ such that $\int \tilde{Y}(\eps) d\mu (\eps) = y$.} where $\mathcal{Y}$ is the set of $\eps$-measurable functions that map to $B$.

In addition, the (integrated) indirect utility function
\[
V(\beta'_1 x_1, \ldots, \beta'_K x_K) = \max_{y \in \overline{B}} \sum_{k = 1}^K y_k (\beta_k' x_k) + \overline{D}(y),
\]
satisfies
\[
V(\beta'_1 x_1, \ldots, \beta'_K x_K) = \int \left( \max_{y \in B} \sum_{k = 1}^K y_k (\beta'_k x_k) + D(y, \eps) \right) d\mu ( \eps).
\]
\end{lemma}
Note that assuming $Y(\cdot)$ satisfies (\ref{eq:latent}) requires the argmax set to be nonempty. This is a behavioral restriction that imposes sufficient structure for the theorem to go through, and imposes minimal restrictions on $D$. In particular, $D$ can be $-\infty$ for certain combinations of $(y,\eps)$ and need not be continuous. This allows us to treat limited consideration models as in Example~\ref{ex:bundles}.

We leverage the aggregation result from Lemma~\ref{l:agg} to use calculus-based techniques for identification. To illustrate how aggregation can lead to smoothness, recall that $Y(x,\beta,\eps)$ in discrete choice is a vector of indicators denoting which good is chosen (assuming no ties). Derivatives with respect to $x$ either do not exist at certain points, or are zero and contain little information. 

We smooth choices by working with
\[
\overline{Y}(x, \beta) := \int Y(x, \beta, \eps) d \mu (\eps)
\]
with $\mu$ as in Assumption~\ref{a:factor}. In discrete choice, when $\eps$ is integrated out $\overline{Y}(x,\beta)$ can be interpreted as the vector of probabilities conditional on only the utility indices. However, this general framework allows us to use the same tools to address discrete and continuous choice. For example, choices could involve a single discrete choice, discrete bundle choice, a prospective matching, continuous quantities of several goods, or time use among other settings.

We places some additional high-level sufficient conditions relative to the conclusions of Lemma~\ref{l:agg}.
\begin{assm} \label{a:maint} Assume the following:
\begin{enumerate}[(i)] 
    \item \[
    \overline{Y}(x, \beta) = \argmax_{y \in \overline{B}} \sum_{k = 1}^K y_k (\beta'_k x_k) + \overline{D}(y).
    \]
    \item $\overline{B} \subseteq \mathbb{R}^K$ is a nonempty, closed, and convex set.
    \item $\overline{D} : \mathbb{R}^K \rightarrow \mathbb{R} \cup \{ - \infty \}$ is concave, upper semi-continuous, and finite at some $y \in \overline{B}$.
\end{enumerate}
\end{assm}
\cite{allen2019identification} provide lower-level conditions that, when combined with Lemma~\ref{eq:latent}, imply this assumption. Part (i) strengthens the conclusion of Lemma~\ref{l:agg} to obtain a unique maximizer. Concavity in part (iii) is milder than it first appears, and delivers no additional restrictions on $\overline{Y}(x,\beta)$ when the other assumptions are maintained. See the discussion in \cite{allen2019identification}.

To further present the foundation of the identification results, we present a version of the envelope theorem.
\begin{lemma} \label{l:env}
Let Assumption~\ref{a:maint} hold. It follows that
\begin{equation} \label{eq:env}
\overline{Y}_k(x, \beta) = \partial_k V(\beta'_1 x_1, \ldots, \beta'_K x_K).
\end{equation}
\end{lemma}
Here, $\overline{Y}_k(x,\beta)$ is the $k$-th component of $\overline{Y}(x,\beta)$ and $\partial_k V(\beta'_1 x_1, \ldots, \beta'_K x_K)$ is the derivative with respect to the $k$-th dimension of $V$ evaluated at the point $(\beta'_1 x_1, \ldots, \beta'_K x_K)'$. We use similar notation for the rest of the paper. Differentiability of $V$ is implied by the fact that $\overline{Y}$ is the unique maximizer. This is the primary implication of Assumption~\ref{a:maint} that we use for this paper. 

\cite{fox2012random} use a structure similar to (\ref{eq:env}), showing that when $V$ is known, it is possible to identify moments of the distribution of $\beta$. We differ because we do not require an analyst to specify $V$. Instead, we require certain moments to be nonzero as a relevance condition. Appendix~\ref{s:vknown} provides further details and a comparison with their approach. A related structure is considered in \cite{lewbel2017unobserved}, who identify the distribution of random coefficients when an analogue of $\partial_k V$ is known in advance or additively separable in arguments. We do not impose this structure.

We also leverage a symmetry property of mixed partial derivatives that results from the optimizing behavior in Assumption~\ref{a:maint}. For a vector of indices $\gamma = (\gamma_1, \ldots, \gamma_M) \in \{1, \ldots, K \}^M$ and a sufficiently differentiable function $f : \mathbb{R}^K \rightarrow \mathbb{R}$, let
\[
\partial_{\gamma} f := \partial_{\gamma_1} \cdots \partial_{\gamma_M} f.
\]
\begin{lemma} \label{l:sym}
Suppose $V$ is $M$-times continuously differentiable in a neighborhood of $\vec{u} \in \mathbb{R}^K$. Let $\gamma, \delta \in \{ 1, \ldots, K \}^{M}$ be vectors of indices in which each index occurs the same number of times in both $\gamma$ and $\delta$. It follows that
\[
\partial_{\gamma} V(\vec{u}) = \partial_{\delta} V(\vec{u}). 
\]
\end{lemma}
This result states that the order in which we take partial derivatives does not matter. For example, when $M = 2$ we have the usual symmetry property of mixed partial derivatives with respect to dimensions $j,k \in \{1,\ldots,K\}$ that
\[
\partial_{j,k} V(\vec{u}) = \partial_{k,j} V(\vec{u}).
\]
The lemma follows by repeated application of the $M = 2$ case.

\section{Main Result}\label{sec:main}
With the foundations in place, we now turn to the task of identifying moments of random coefficients. We focus on conditions where certain $M$-th order moments of the distribution of $\beta$ are identified. In particular, if the assumptions hold for \textit{all} $M$, then all moments of the distribution of random coefficients are identified.

We assume regressors are continuous and satisfy an exclusion restriction. 
\begin{assm} \label{a:cont}
All covariates are continuous. In addition, each $x_k$ is a vector of regressors specific to the $k$-th good.
\end{assm}

We now provide some intuition for the main result (Theorem~\ref{thm:main}). We consider identifying second moments of $\beta$ ($M = 2$) when there are two goods ($K = 2$) and each good has a single covariate ($d_k = 1$). We focus on second moments since this example captures the power of the results in the simplest non-trivial setting. We write the partial derivative of a function, $f$, with respect to the covariates of the $j$-th good, $x_j$, as $\partial_{x_j}f$. Differentiating the envelope theorem (Lemma~\ref{l:env}) and evaluating at $x = 0$ we obtain
\[
\partial_{x_j} \overline{Y}_k(0, \beta) = \partial_{j,k} V(0) \beta_j.\footnote{Here we abuse notation and for the function $f$, we let $\partial_{s} f(0) = \left. \partial_{s} f(z)  \right|_{z=0}$.}
\]
This uses the fact that $x_j$ is continuous and excluded from the utility index of other goods. This can be repeated with other mixed partial derivatives. Importantly, by evaluating derivatives at the point $x = 0$, the terms $\partial_{j,k}V(0)$ do not depend on $\beta$. Thus, when integrating over the values of the random coefficients, the term involving $V$ passes outside of the integral. In particular, integrating over $\beta$ yields the following system of equations
\begin{equation} \label{eq:example}
\begin{split}
    \partial_{x_1} \partial_{x_1} \overline{Y}_2(0) & = \partial_{1,1,2} V(0) \int \beta^2_1 d\nu(\beta) \\
    \partial_{x_1} \partial_{x_2} \overline{Y}_2(0) & = \partial_{1,2,2} V(0) \int \beta_2 \beta_1 d\nu( \beta) \\
    \partial_{x_2} \partial_{x_1} \overline{Y}_1(0) & = \partial_{2,1,1} V(0) \int \beta_1 \beta_2 d\nu( \beta) \\
    \partial_{x_2} \partial_{x_2} \overline{Y}_1(0) & = \partial_{2,2,1} V(0) \int \beta^2_2 d\nu ( \beta)
\end{split}
\end{equation}
where we have implicitly assumed that differentiation and integration can be interchanged.

Assume that the derivatives of $\overline{Y}$ are identified. At first glance, this is a system of four equations with seven unknowns (clearly the $\beta_1 \beta_2$ and $\beta_2 \beta_1$ moments are equal). However, when $V$ is sufficiently differentiable, partial derivatives of $V$ do not depend on the order of differentiation (Lemma~\ref{l:sym}), which eliminates two unknowns. Using a scale assumption that $\int \beta^2_1 d\nu (\beta)$ is known \textit{a priori} will eliminate an unknown and gives a system with $4$ equations and $4$ unknowns. We show that this is enough to identify all second moments of $\beta$.

To constructively see how the moments are identified, note that using symmetry of derivatives, the first and third equations identify $\int \beta_1 \beta_2 d\nu ( \beta)$. Using this, we identify $\partial_{1,2,2} V(0)$ using the second equation. Again using symmetry of derivatives and combining this with the last equation identifies $\int \beta^2_2 d\nu (\beta)$. Once all moments are identified, the remaining third order derivatives of $V$ can be identified at $0$.\footnote{This part also requires the equations
\begin{align*}
    \partial_{x_1} \partial_{x_1} \overline{Y}_1(0) & = \partial_{1,1,1} V(0) \int \beta^2_1 d\nu( \beta) \\
    \partial_{x_2} \partial_{x_2} \overline{Y}_2(0) & = \partial_{2,2,2} V(0) \int \beta^2_2 d\nu (\beta)
\end{align*}
to identify $\partial_{1,1,1} V(0)$ and $\partial_{2,2,2} V(0)$.} 

We now provide formal conditions that justify the intuitive argument for any number of goods, covariates, and order of moment $M$.
\begin{assm} \label{a:scale}
For the natural number $M$, $\int \beta^M_{1,1} d\nu (\beta)$ is finite, known \textit{a priori}, and nonzero.
\end{assm}
Assumption~\ref{a:scale} holds if we set $\beta_{1,1} = 1$, for example, but is considerably more general. It allows heterogeneity in the sign of $\beta_{1,1}$, for example. In general, if one wants to identify all moments of $\beta$ using the main result, then for every $M$ Assumption~\ref{a:scale} must hold. This assumption holds when the distribution of $\beta_{1,1}$ is known \textit{a priori} and the distribution has nonzero moments of all orders. If Assumption~\ref{a:scale} is dropped, the results in this paper establish identification of the ratio of any nonzero $M$-th order moments. Thus, Assumption~\ref{a:scale} can be appropriately modified by instead holding fixed the value of some other nonzero $M$-th order moment of the form $\int \beta_{k_1, \ell_1} \cdots \beta_{k_M, \ell_M} d \nu (\beta)$. We show in Section~\ref{sec:alt} that if $\beta_{1,1}$ is independent of all other components of $\beta$, then identification is possible using a single scale assumption on the first moment.

Recall that with minor abuse of notation we set
\[
\overline{Y}(x) = \int \overline{Y}(x, \beta) d\nu (\beta).
\]
We require the following regularity conditions.

\begin{assm} \label{a:reg}
For the natural number $M$, the following conditions hold:
\begin{enumerate}[(i)]
    \item For each good $k$, one can interchange integration and differentiation for all $M$-th order partial derivatives at $x=0$ so that 
    \[
    \partial_{x_{k_1,\ell_1}} \cdots \partial_{x_{k_M,\ell_M}}     \overline{Y}_k(0) = \int \partial_{x_{k_1,\ell_1}} \cdots     \partial_{x_{k_M,\ell_M}} \overline{Y}_k(0, \beta) d \nu ( \beta)
    \]
    holds.
    \item Each $M$-th order moment
    \[
    \int \beta_{k_1,\ell_1} \cdots \beta_{k_M, \ell_M} d \nu (\beta)
    \]
    exists and is finite.
    \item $V$ is $(M+1)$-times continuously differentiable in a neighborhood of $0$.
    \item For each $\gamma \in \{1, \ldots, K \}^{M+1}$,
    \[
    \partial_{\gamma} V (0) \neq 0.
    \]
    \item $\overline{Y}(x)$ is known in a neighborhood of $x = 0$, or more generally it is known in a neighborhood of $x = 0$ with respect to the weakly positive orthant of $\mathbb{R}^{\sum_{k = 1}^K d_k}$.\footnote{More formally, the second condition can be written as follows: for some neighborhood $H$ of $x = 0$ in the usual topology on $\mathbb{R}^{\sum_{k = 1}^K d_k}$, $\overline{Y}(x)$ is known on $H \cap \mathbb{R}^{\sum_{k = 1}^K d_k}_{+}$, where $\mathbb{R}_{+} = \mathbb{R} \cap [0, \infty)$.}
\end{enumerate}
\end{assm}
These regularity conditions parallel assumptions in \cite{fox2012random}. To interpret part (i), note that $\nu$ can be a discrete probability measure over $\beta$ with finite support. For discrete measures, (i) holds whenever $\overline{Y}_k(x,\beta)$ is $M$-times differentiable in $x$ for every $\beta$ in its support. Part (ii) formalizes that the moments we wish to identify exist and are finite.

Parts (iii) and (iv) can be linked to derivatives of the function $\overline{Y}(x)$ via the envelope theorem (Lemma~\ref{l:env}). Indeed, differentiating the envelope theorem for the $k$-th good, evaluating the derivative of $\bar{Y}_k$ with respect to $x_{j,\ell}$ at $x = 0$, and taking expectations yields
\begin{align} \label{eq:envassm}
\frac{\partial {\overline{Y}_k(0)}}{\partial x_{j, \ell}} = \partial_{j,k} V(0) \int \beta_{j,\ell} d\nu( \beta).
\end{align} 
Thus, when $V$ is $(M+1)$-times continuously differentiable, it follows that $\overline{Y}$ is $M$-times continuously differentiable. Moreover, if one sees empirically that $\frac{\partial {\overline{Y}_k(0)}}{\partial x_{j,\ell}} \neq 0$, then it follows that that $\partial_{j,k} V(0) \neq 0$ (whenever this derivative exists).

\cite{fox2012random} show that condition (iv) holds for random coefficients logit, for ``most'' values of nonrandom intercepts. Specifically, the set of intercepts that violate (iv) for some $\gamma$ has Lebesgue measure $0$. In general, whether (iv) holds depends on features of the distribution of $\eps$ and choice of $D$, which are example specific. For example, part (iv) rules out pure characteristic discrete choice models as in \cite{berry2007pure} and \cite{dunker2017nonparametric}. These models do not include a random intercept, and so the value function for the pure characteristics model, $V^{PC}$, can be written as
\[
V^{PC}(\beta'_1 x_1, \ldots, \beta'_K x_K) = \sup_{y \in \overline{B}} \sum_{k = 1}^K y_k (\beta_k'x_k)
\]
without the additive disturbance $\overline{D}$, where $\overline{B}$ is the probability simplex. $V^{PC}$ does not have a non-zero derivative at $x=0$ for any $M \ge 1$ with this constraint set. This choice of $V^{PC}$ also does not always induce a unique maximizer.

More generally, condition (iv) requires that goods in the demand system are related. For example, if the original $K$-good demand system can be written as $K$ separate $1$-good demand systems, then derivatives of the form $V_{j,k}(0)$ will be zero for $j \neq k$. This is because under this separability assumption, the utility index of good $j$ does not alter the demand for the $k$-th good. In general, (iv) cannot be relaxed for the main result to hold without additional assumptions. For example, if we impose that $\beta_j = \beta_k$ (a.s.) for all $j,k \in \{1,\ldots,K\}$, then one can identify ratios of moments under the weaker assumption that $\partial^{M+1}_j V(0) \neq 0$ for some $j$. See Appendix~\ref{supp:cherno}.

Condition (v) states that $\overline{Y}(x)$ is identified over a small region near $x = 0$. The constructive identification results in \cite{fox2012random} and \cite{chernozhukov2019nonseparable} have also made use of variation around zero. In contrast, most of the literature instead requires identification of $\overline{Y}$ either for all $x$ or for a set over which $x$ is unbounded along some dimensions.

To interpret condition (v), suppose that $X$, $\beta$, and $\eps$ are all independent, and we identify $\overline{Y}$ from a continuous version of the conditional mean of $Y$ given $X$. For this case, condition (v) is implied when the support of $X$ contains an open ball around $x = 0$. The second more general part of (v) highlights that the results also apply when the average structural function is identified over a weakly positive region. Thus, our results do not rule out prices. We can handle this case because we only need to identify certain derivatives of $\overline{Y}$ at $0$. These derivatives of $\overline{Y}$ at $0$ are identified in this case by calculating derivatives from ``one-sided'' limits involving non-negative numbers. 

The final assumption used for identification is that a sufficiently rich set of $M$-th order moments of $\beta$ are nonzero.
\begin{assm} \label{a:exist}
For the natural number $M$ and each tuple of good indices $(k_1, \ldots, k_M) \in \{1, \ldots, K \}^M$, there is a corresponding tuple of characteristic indices $(\ell_1, \ldots, \ell_M) \in \prod_{m=1}^M \{1,\ldots,d_{k_m} \}$ such that the $M$-th order moment
\[
\int \beta_{k_1,\ell_1} \cdots \beta_{k_M, \ell_M}\nu (d \beta)
\]
exists and is nonzero.
\end{assm}
This is a relevance condition. It is not necessary to know which indices $(\ell_1, \ldots, \ell_M)$ satisfy this condition in advance.\footnote{See the discussion after Lemma~\ref{l:a2} in Appendix~\ref{app:thm1}.} A sufficient condition for this is that for every $k$-th good there is a regressor $\ell_k \in \{1,\ldots,d_k\}$ such that either $\beta_{k,\ell_k} \geq 0$ almost surely or $\beta_{k,\ell_k} \leq 0$ almost surely, with positive probability that the inequality is strict. A stronger condition that implies this is that $\beta_{k,1} = 1$ (a.s.) for every $k$-th good by setting $\ell_1 = \cdots = \ell_M = 1$. This is a common assumption in the literature \citep{berry2009nonparametric,briesch2010nonparametric,dunker2017nonparametric}. However, \cite{ichimura1998maximum} and \cite{gautier2013nonparametric} establish identification of random coefficients models for binary discrete choice using a more general halfspace condition.  

With these assumptions, we can now state the main result of the paper.
\begin{theorem} \label{thm:main}
Let Assumptions~\ref{a:factor}-\ref{a:exist} hold with the same natural number $M$. Each $M$-th order moment of the form
\[
\int \beta_{k_1,\ell_1} \cdots \beta_{k_M, \ell_M} d\nu (\beta)
\]
is identified. In addition, for each $\gamma \in \{1, \ldots, K \}^{M+1}$,
\[
\partial_{\gamma} V (0)
\]
is identified.
\end{theorem}
This result establishes nonparametric identification of certain moments of $\beta$. It can be directly used to establish semiparametric identification of the distribution of $\beta$ for certain parametric families without specifying other objects (e.g. $V$). For example, if $\beta$ is normally distributed then Theorem~\ref{thm:main} identifies the distribution when the assumptions hold for $M \in \{1, 2\}$ because normal distributions are characterized by means and covariances. Recall that while we identify non-centered moments, we can use this information to identify centered moments. More generally, for any distribution of $\beta$ that is defined by its moments up to order $M$ this result estabilishes identification of the distribution.

\begin{cor} \label{cor:id}
Suppose Assumptions~\ref{a:factor}-\ref{a:exist} hold for each $M\le \overline{M}$ (which could be $\infty$) and the distribution of $\beta$ is determined by its first $\overline{M}$ moments. It follows that the distribution of $\beta$ is identified.
\end{cor}
\cite{fox2012random} and \cite{il2014identification} describe a sufficient condition for a distribution to be determined by its moments. Distributions with compact finite support are determined by their moments. Lognormal distributions are an example of distributions that are not determined by integer moments \citep{heyde1963property}. That is, there are other nonparametric distributions that can match the same moments. However, the parameters may still be identified within the lognormal class.

\begin{remark}[Constructive Identification]
The proof of Theorem~\ref{thm:main} is constructive. While a detailed approach to estimation is beyond the scope of this paper, the constructive results can be used to show how consistent estimation of the ratios of certain derivatives allows one to consistently estimate moments of the distribution of random coefficients. See  Appendix~\ref{app:estimation} for a brief outline. Estimation error of the $M$-th order moments can be bounded when the analyst has a suitable estimator of the $M$-th order derivatives of the average structural function. For example, \cite{chen2018optimal} provide an approach to estimate derivatives of average structural functions in the presence of endogeneity.
\end{remark}

\subsection{Independence of $\beta_{1,1}$} \label{sec:alt}
Identifying the distribution of $\beta$ using Corollary~\ref{cor:id} requires Assumption~\ref{a:scale}, which specifies all moments of the form $\int \beta^M_{1,1} d\nu (\beta)$. When the distribution of $\beta$ is identified from its moments, one must specify the marginal distribution of $\beta_{1,1}$ in advance to apply Corollary~\ref{cor:id}. While the common assumption $\beta_{1,1} = 1$ (a.s.) implies Assumption~\ref{a:scale}, one may not want to impose either this assumption or the weaker assumption that $\beta_{1,1}$ has a known distribution. This section describes an alternative assumption that ensures identification of moments of $\beta$. In particular, we assume $\beta_{1,1}$ is independent of other components of $\beta$. 

With this independence assumption, we show that a single scale assumption on the first moment of $\beta_{1,1}$ allows us to identify a rich collection of moments. This contrasts with Theorem~\ref{thm:main}, which uses an assumption on the $M$-th order moment of $\beta_{1,1}$ to identify \emph{only} $M$-th order moments of $\beta$.
\begin{assm} \label{a:oneind}
For the measure $\nu$ as defined in Assumption~\ref{a:factor}, $\beta_{1,1}$ is independent of all other components of $\beta$. In addition, $| \int \beta_{1,1} d \nu (\beta) |$ is finite, known \textit{a priori}, and nonzero.
\end{assm}
Alternatively, one could set the absolute value of some other order moment of $\beta_{1,1}$, but we focus on the first moment since it facilitates interpretation. Independence between $\beta_{1,1}$ and other components is considerably weaker than assuming $\beta_{1,1} = 1$ almost surely. For example, this allows $\beta_{1,1}$ to be sometimes negative and sometimes positive. Thus, different individuals can be repelled or attracted to higher values of $x_{1,1}$.

Replacing Assumption~\ref{a:scale} with Assumption~\ref{a:oneind}, we obtain the following counterpart of Theorem~\ref{prop:alt}.
\begin{prop} \label{prop:alt}
Let $K\ge 2$ and Assumptions~\ref{a:factor}-\ref{a:cont} and~\ref{a:reg}-\ref{a:oneind} hold for all for each $M\le \overline{M}$ (which could be $\infty$). It follows that the $M$-th order moment of the form
\[
\int \beta_{k_1,\ell_1} \cdots \beta_{k_M, \ell_M}\nu (d \beta)
\]
is identified. In addition, for each $\gamma \in \{1, \ldots, K \}^{M+1}$,
\[
\partial_{\gamma} V (0)
\]
is identified.

If $K = 1$, the same conclusions hold given the additional assumption that for each each $M \leq \overline{M}$, there exists an order $M-1$ moment such that
\[
\int \beta_{1,\ell_{1}} \cdots \beta_{1, \ell_{M-1}} d\nu(\beta) \neq 0,
\]
where $\ell_{m} \neq 1$ for every $m \in \{1, \ldots, M - 1\}$.
\end{prop}

Relative to Theorem~\ref{thm:main}, independence of $\beta_{1,1}$ from the other components allows us to relate the $M$-th and $M-1$ order moments. To see this, consider some $M$-th order moment in which $\beta_{1,1}$ appears exactly once. Using independence, we obtain
\[
\int \beta_{1,1} \beta_{k_2, \ell_2} \cdots \beta_{k_M, \ell_M} d\nu (\beta) = \int \beta_{1,1} d\nu (\beta) \int \beta_{k_2, \ell_2} \cdots \beta_{k_M, \ell_M} d\nu (\beta). 
\]
In the proof, we show that $\int \beta_{1,1} d\nu (\beta)$ can be identified when $| \int \beta_{1,1} d\nu (\beta)|$ is finite, known \textit{a priori}, and non-zero. With this knowledge, we can identify the ratio of all $M$-th order moments to all $(M-1)$-th order moments and apply induction to identify all $M \le \bar{M}$ order moments.

\begin{remark}[One good]
When $K \geq 2$, the assumptions of Proposition~\ref{prop:alt} impose that there are multiple relevant characteristics. The additional assumption in the $K = 1$ case is a relevance condition on a characteristic other than $x_{1,1}$.
\end{remark}

\begin{remark}[``Normalizations'']
Provided $\int \beta_{1,1} \nu (d \beta)$ exists and is nonzero, it is a normalization to set $| \int \beta_{1,1} \nu (d \beta) | = 1$. In other words, this imposes no additional restrictions on the model. This is seen by noting that if we divide the original latent utility model by $| \int \beta_{1,1} \nu (d \beta) |$, then the argmax set does not change and none of the assumptions in Proposition~\ref{prop:alt} are affected by this division.

A natural intuition is that when $\beta_{1,1} > 0$ almost surely, it is also a normalization to divide the latent utility by $\beta_{1,1}$ and rewrite the problem with $\beta_{1,1} = 1$. This is true if we only inspect the original latent utility model (Equation~\ref{eq:latent}), but is no longer true when we consider independence or certain other additional assumptions on the integrated model in Assumption~\ref{a:maint}. Recall that Assumption~\ref{a:factor} means $\beta$ and $\eps$ are independent in the definition of the average structural function $\overline{Y}$. In general, this assumption is not invariant to division by $\beta_{1,1}$. This means that setting $\beta_{1,1} = 1$ (a.s.) provides additional restrictions relative to the assumptions of Proposition~\ref{prop:alt}.
\end{remark}

\section{Welfare Analysis and Counterfactuals}\label{sec:welfare}

We now turn to identification of certain welfare and counterfactual objects. Identification is established given identification of certain features of $V$, which is the indirect utility function obtained when random intercepts are integrated out. We first provide three results that identify differences in $V$. Using these results, we discuss welfare analysis and counterfactuals.

The reason we identify $V$ is that we can use the envelope theorem to determine certain average choices
\[
\overline{Y}(x,\beta) = \nabla V(\beta'_1 x_1, \ldots, \beta'_K x_K).
\]
We require identification of the right hand side at values other than $0$ to consider counterfactuals at new values of covariates.

\subsection{Identification of $V$}

We first provide conditions under which identification of partial derivatives of $V$ at $0$ allows us to directly extrapolate the function. Specifically, we assume $V$ is a real analytic function. That is, $V$ has derivatives of all orders and agrees with its Taylor series in a neighborhood of every point. Real analytic functions have the important property that local information can be used to reconstruct the function globally by extrapolating. This is similar to common parametric classes of functions. However, the set of real analytic functions is infinite dimensional.

\begin{cor} \label{cor:vid}
Let the assumptions of Theorem~\ref{thm:main} or the assumptions of Proposition~\ref{prop:alt} hold with $\overline{M} = \infty$. If $V$ is a real analytic function, then it is identified up to an additive constant.
\end{cor}

One way to drop the assumption that $V$ is a real analytic function is to instead assume $\beta_{k,1} = 1$ almost surely for each $k$. With this assumption, let $\tilde{x}$ be a value that is zero for every characteristic except the first characteristic of each good. Then the envelope theorem (Lemma~\ref{l:env}) specializes to
\[
\overline{Y}_k(\tilde{x},\beta) = \partial_k V(\tilde{x}_{1,1}, \ldots, \tilde{x}_{K,1}).
\]
This does not depend on $\beta$, and so by taking expectations, the average structural function identifies the derivative of $V$ at the point $(\tilde{x}_{1,1}, \ldots, \tilde{x}_{K,1})$. By integrating the derivatives we can identify differences in $V$, as we now formalize.

\begin{prop} \label{prop:idhomog}
Let Assumption~\ref{a:maint} hold and assume $\beta_{k,1} = 1$ almost surely for each $k\in \{1,\ldots,K\}$. Suppose $\overline{Y}(x)$ is identified for all $x = (x_1', \ldots, x_K')'$ satisfying $x_{k,1} \in [\underline{x}_{k,1}, \overline{x}_{k,1}]$ for each $k$, and $x_{k,j} = 0$ for $j > 1$. Then differences in $V$ are identified over the region $\times_{k = 1}^K [\underline{x}_{k,1}, \overline{x}_{k,1}]$. In particular, if $\underline{x}_{k,1} = -\infty$ and $\overline{x}_{k,1} = \infty$ for each $k$, then $V$ is identified up to an additive constant.
\end{prop}

The results on welfare and counterfactual analysis require that derivatives of $V$ be identified at certain values $(\beta'_1 x_1, \cdots, \beta'_K x_K)$. If the support of $\beta$ is compact, then it is not necessary to identify $V$ everywhere, and so it is not necessary to have $\underline{x}_{k,1} = -\infty$ and $\overline{x}_{k,1} = \infty$ to apply Proposition~\ref{prop:idhomog}.

Finally, we mention a third way to identify differences in $V$. A key distinction is that it requires identification of the distribution of $\overline{Y}(x,\beta)$ for fixed $x$, rather than identification of the average structural function as in the rest of the paper. We adapt the following lemma.
\begin{lemma}[\cite{mccann1995existence}; statement from \cite{chernozhukov2019single}, Corollary 2] \label{lem:chern}
Let $W = f(\eta)$, where $W$ and $\eta$ have the same finite dimension. Suppose $f$ is the gradient of a convex function, $\eta$ has a known distribution that is absolutely continuous with respect to Lesbesgue measure, the distribution of $W$ is known, and $\eta$ and $W$ have finite variance. It follows that $f$ is identified.
\end{lemma}
This result can be applied to our setting by adapting the envelope theorem,
\[
\overline{Y} (x,\beta) = \nabla V(\beta_1 ' x_1, \ldots, \beta'_K x_K).
\]
Interpret $W = \overline{Y}(x,\beta)$, $f = \nabla V$, and $\eta = (\beta_1' x_1, \ldots, \beta_K' x_K)$. When $x$ is fixed and the distribution of $\beta$ is identified (from previous arguments), the distribution of $\eta$ is known. The function $V$ is convex, and so the lemma provides conditions under which $\nabla V$ is identified. Importantly, the lemma can be applied at a \textit{single} $x$, so it is not necessary to have full support of covariates to apply the result. Such an $x$ cannot be arbitrary. For example, when $x = 0$ the distribution corresponding to $\eta$ is not absolutely continuous. The lemma can still be applied for $x$ near but not equal to $0$. Moreover, to apply the lemma, the distribution of $\beta$ cannot be degenerate, i.e. there must be truly ``random'' coefficients. If $\beta$ is almost surely equal to a constant, then the distribution corresponding to $\eta$ is not absolutely continuous and the lemma does not apply.

Importantly, to apply Lemma~\ref{lem:chern} in our setting, the distribution of $\overline{Y}(x,\beta)$ must be identified at some fixed $x$. One example in which this lemma can be applied is when $\eps$ in the original latent utility model is not present, so that $\overline{Y}(x,\beta)$ corresponds to the observable choices given $x$ and $\beta$. Such structure could be appropriate in a continuous choice model in which all unobservable heterogeneity is controlled by the random slopes $\beta$, and in which the choices (rather than e.g. average choices for a group of individuals) are observed.

\subsection{Welfare}
We now describe how identification of $V$ leads to identification of certain welfare objects. First, recall that $V$ may be interpreted as the indirect utility conditional on the utility index (Lemma~\ref{l:agg}) where the random intercept $\eps$ under the measure $\mu$ is integrated out. To interpret $V(\cdot)$ as a welfare object, suppose $\beta$ is an individual-specific term that is random across the population but constant across decisions for the same individual. Interpret the random intercept $\eps$ as an idiosyncratic taste shock across decision problems. Then $V(\beta'_1 x_1, \ldots, \beta'_K x_K)$ is an individual-specific (integrated) indirect utility. The conditions of Corollary~\ref{cor:id} identify the distribution of $V(\beta'_1 x_1, \ldots, \beta'_K x_K)$ under the measure $\nu$ up to an additive constant once the values of covariates are fixed.

Thus, we can identify the distribution of individual-specific indirect utilities. By further integrating out the distribution of $\beta$, we also identify differences in the average indirect utility via Lemma~\ref{l:agg}:
\begin{align} \label{eq:weirdagg}
\int V (\beta'_1 x_1, \ldots, \beta'_K x_K) d\nu ( \beta)
= \int \int \sup_{y \in B} \sum_{k = 1}^K y_k (\beta'_k x_k) + D(y, \eps) d\mu (\eps) d\nu (\beta).
\end{align}
This result holds regardless of whether the distribution of $\eps$ is identified since $V$ is a welfare-relevant summary measure of the distribution of $\eps$. Indeed, we do not establish identification of the distribution of
\[
\sup_{y \in B} \sum_{k = 1}^K y_k (\beta'_k x_k) + D(y, \eps)
\]
according to the product measure $\mu \times \nu$ over $(\beta,\eps)$. In particular, since this paper does not study identification of $\eps$, we do not identify the distribution of indirect utilities including random intercepts.

Note that the units of (\ref{eq:weirdagg}) are relative to the scale assumption used to identify the distribution of $\beta$. If we impose the scale assumptions in Theorem~\ref{thm:main} to apply Corollary~\ref{cor:vid}, then the distribution of $\beta_{1,1}$ is fixed. Thus, the units of Equation~\ref{eq:weirdagg} are set by the \textit{distribution} of the conversion rate between $x_{1,1}$ and utils. In contrast, if we impose the conditions of Proposition~\ref{prop:alt}, then the scale is determined by $\left| \int \beta_{1,1}  \nu (\beta) \right|$. For this case, the units of Equation~\ref{eq:weirdagg} are relative to the \textit{average} conversion rate between $x_{1,1}$ and utils. 

An alternative measure of average indirect utility is
\begin{align} \label{eq:constantagg}
\int \frac{1}{|\beta_{1,1}|} V (\beta'_1 x_1, \ldots, \beta'_K x_K) d\nu ( \beta)
= \int \int \frac{1}{|\beta_{1,1}|} \sup_{y \in B} \sum_{k = 1}^K y_k (\beta'_k x_k) + D(y, \eps) d\mu (\eps) d\nu (\beta).
\end{align}
This sets the conversion rate of $x_{1,1}$ and utils to $\pm 1$. Importantly, this preserves whether the first characteristic is desirable or undesirable. It also forces the intensity of preference to be constant across individuals. This welfare measure is most interpretable when the regressor has a homogeneous sign. For example, if $x_{1,1}$ is the (negative) price of good $1$, then $\beta_{1,1} < 0$ is a natural assumption, and the units of Equation~\ref{eq:constantagg} are in dollars.

\subsection{Counterfactuals}
Once $V$ is identified, we can also answer certain counterfactual questions involving quantities at new values of covariates. To this end, recall from Lemma~\ref{l:env} that
\begin{equation} \label{eq:lem2restate}
\overline{Y}_k(x, \beta) = \partial_k V(\beta'_1 x_1, \ldots, \beta'_K x_K).
\end{equation}
Here, $\overline{Y}_k(x, \beta)$ is the demand for good $k$ fixing covariates and the random intercept, but integrating out the distribution of $\eps$. We interpret $\beta$ as an individual-specific parameter that is constant across decision problems, while $\eps$ is an idiosyncratic shock that can vary across decision problems. Thus, $\overline{Y}_k(x, \beta)$ is the individual-specific average quantity of the $k$-th good. Once $V$ and the distribution of $\beta$ are identified, we can identify the distribution of $\overline{Y}_k(x,\cdot)$ from Equation~\ref{eq:lem2restate}.

Conceptually, this shows it is possible to start with identification of the average structural function (``mean choices'') around $x = 0$ to identify the integrated choices $\overline{Y}_k(x,\beta)$ at all values of the covariates. This also implies that any value $x$ at which $\overline{Y}$ can be identified directly from data (as opposed to the theoretical analysis just described) provides overidentifying information.

\section{Discussion}\label{sec:discussion}

We now provide additional discussion of the main results in the paper.\\

\begin{remark}[Location of Taste Homogeneity]
Our analysis can be applied to settings in which covariates do not vary around zero. For example, we could recenter via $\tilde{X} = X - \E[X]$ so that identification uses variation in the average structural function around the mean, rather than variation around $0$. This is noted as well in \cite{fox2012random}. Importantly, the assumptions in this paper are not typically invariant to recentering. Thus, the assumptions must be made \textit{given} a particular centering at which the average structural function is identified. In addition, the location of recentering defines where the random slopes do not alter preferences since $\beta'x = 0$ when $x = 0$. In words, the choice of centering sets a location of taste homogeneity with regard to the random slopes (but not random intercepts).\\
\end{remark}

\begin{remark}[Nonlinear Random Coefficients]
The results may be adapted to certain models in which coefficients are not linear. Suppose instead of the linear index $\beta'_k x_k$, we have $x^{\rho_k}_k$ for a scalar shifter $x_k$ and random exponent $\rho_k$. Applying the envelope theorem yields
\[
\partial_{x_j} \overline{Y}_k(x, \rho) = \partial_{j,k} V(x_1^{\rho_1}, \ldots, x^{\rho_K}_K) \rho_j x^{\rho_j - 1}_j
\]
where $\rho = (\rho_1, \ldots, \rho_K)$ collects the random exponents. In particular, the partial derivatives of $V$ can be evaluated at a vector of ones so that $x^{\rho_j - 1}_j = 1$ for any random exponent. By taking expectations with respect to $\rho$, evaluating the above equation around covariates equal to one, and using symmetry of mixed partial derivatives, we obtain
\[
\frac{\partial_{x_j} \overline{Y}_k(1)}{\partial_{x_k} \overline{Y}_j(1) } = \frac{\E[ \rho_j ]}{ \E[ \rho_k ]}.
\]
Ratios of higher order moments can be identified by considering additional derivatives of the average structural function, and by using symmetry of mixed partials of $V$ evaluated at the vector of ones.\\
\end{remark}

\begin{remark}[Complementarity and Derivatives of $V$]
Recall that the envelope theorem yields
\[
\partial_{x_j} \overline{Y}_k(0, \beta) = \partial_{j,k} V(0) \beta_j.
\]
Thus, second order mixed partial derivatives of $V$ describe how changes in the utility index of good $x_j$ alter the demand for good $k$. The sign of $\partial_{j,k} V(0)$ describes whether goods are local complements ($\partial_{j,k} V(0) \geq 0$) or substitutes ($\partial_{j,k} \leq V(0)$). Theorem~\ref{thm:main} provides conditions under which this derivative is identified, and thus we obtain information on complementarity/substitability with random coefficients. Several papers have studied complementarity in the bundles model (Example~\ref{ex:bundles}) when characteristics shift the marginal utility of a good homogenously (\cite{gentzkow2007,fox2015,chernozhukov2015constrained,hicks,bundles}). To our knowledge, the only paper that studies identification of complementarity in the bundles model with heterogeous tastes for characteristics is \cite{dunker2017nonparametric}. Rather than working with only the average structural function (``mean demands''), \cite{dunker2017nonparametric} require identification of the frequencies that goods are chosen together.\\
\end{remark}

\begin{remark}[Homogeneity of Coefficients]
We do not require the assumption that coefficients are the same across goods, $\beta_j = \beta_k$. One interpretation of this assumption is that preferences are driven by observable characteristics \citep{gorman1980possible,lancaster1966new}, not the label of the good ($j$ vs $k$). In this paper, we assume that the shifters associated with $\beta_k$ vary for good $k$. Thus, setting $\beta_j = \beta_k$ means the shifters that vary for good $j$ are the same as for good $k$. This assumption is inconsistent with some empirical settings, especially outside of discrete choice. For example, it is not satisfied for the bundles model of \cite{gentzkow2007} where internet speed varies for online news but not print news.

Placing restrictions relating coefficients across different goods allows one to either weaken the conditions used for identification or use alternative techniques. See Appendix~\ref{supp:cherno} for additional discussion and relation to \cite{chernozhukov2019nonseparable}.\\
\end{remark}

\begin{remark}[Testability]
The conditions of Theorem~\ref{thm:main} imply testable implications because of theoretical relationships between different moments. To see this, we revisit the system of equations (\ref{eq:example}) used previously to illustrate the identification technique. Dividing the first and third equations, and the second and fourth equations, we obtain
\begin{align*}
 \frac{\partial_{x_1} \partial_{x_1} \overline{Y}_2(0)}{\partial_{x_2} \partial_{x_1} \overline{Y}_1(0)} & = \frac{\int \beta^2_1 \nu(d \beta)}{\int \beta_1 \beta_2 \nu(d \beta)}
 & \frac{\partial_{x_2} \partial_{x_2} \overline{Y}_1(0)}{\partial_{x_1} \partial_{x_2} \overline{Y}_2(0)} & = \frac{\int \beta^2_2 \nu(d \beta)}{\int \beta_2 \beta_1 \nu(d \beta)}.
\end{align*}
Multiplying these equations and using the Cauchy-Schwarz inequality yields the testable restriction
\[
 \frac{\partial_{x_1} \partial_{x_1} \overline{Y}_2(0)}{\partial_{x_2} \partial_{x_1} \overline{Y}_1(0)} \cdot \frac{\partial_{x_2} \partial_{x_2} \overline{Y}_1(0)}{\partial_{x_1} \partial_{x_2} \overline{Y}_2(0)} \geq 1.
\]
Note that this inequality only concerns the average structural function close to $0$.
\end{remark}

\begin{remark}[Determinacy by Moments]
The results in this paper establish identification of certain moments of the distribution of random coefficients. It is natural to wonder whether the distribution can be identified without requiring that it be uniquely determined by its moments. A related question has been studied previously. Specifically, consider the setting
\[
Y = A + B'Z,
\]
where $Z$ is independent of $(A,B)$ and the distribution of $(Y,Z)$ is identified. Building on \cite{belisle1997probability}, \cite{masten2017random} shows that, if the support of $Z$ is bounded, then identification of either $(A,B)$ or even the marginal distributions $B_k$ requires that the distribution $(A,B)$ is determined by its moments. In light of this, it is possible that given only the conditions of Theorem~\ref{thm:main} for each $M$, identification of the distribution of $\beta$ requires that it be determined by its moments. Recall that we only assume identification of the average structural function $\overline{Y}(x)$ in a bounded set around $0$.
\end{remark}

\bibliographystyle{plainnat}
\bibliography{ref}

\newpage

\begin{appendices}	

\counterwithin{theorem}{section}
\counterwithin{prop}{section}
\counterwithin{lemma}{section}
\counterwithin{cor}{section}
\counterwithin{assm}{section}
\counterwithin{defn}{section}
\counterwithin{remark}{section}

\section{Proofs of Main Results}

\subsection{Preliminary Lemmas}

\begin{proof}[Proof of Lemma~\ref{l:agg}]
This follows line by line from the proof of \cite{allen2019identification}, Theorem 1. The statement of that result included an additional Assumption 1, which was not used in the proof as long as the underlying choice is appropriately measurable. Here, we start with $Y$ of the form $Y(X,\beta,\eps)$, which is automatically a measurable function.
\end{proof}

\begin{proof}[Proof of Lemma~\ref{l:env}]
See \cite{allen2019identification}, Lemma 1. The result may also be directly proven from \cite{rockafellar1970convex}, Theorems 23.5 and 25.1.
\end{proof}

\begin{proof}[Proof of Lemma~\ref{l:sym}]
The function $V$ is convex. The result then follows from \cite{rockafellar1970convex}, Theorem 4.5, plus repeated differentiation.
\end{proof}

\subsection{Proof of Theorem~\ref{thm:main}} \label{app:thm1}

The following lemmas maintain the assumptions of Theorem~\ref{thm:main}. These assumptions ensure the requisite smoothness assumptions and ensure that the following arguments do not divide by zero.

In order to simplify presentation, we require some additional notation. Let $(\gamma, \xi)$ be a tuple with $\gamma \in \{ 1, \ldots, K \}^M$ denoting good indices, and let $\xi_k \in \{1, \ldots d_{\gamma_k} \}$ describe which characteristic corresponds to the $\gamma_k$-th good.
We set
\[
\partial_{(\gamma,\xi)} \overline{Y}_k(0,\beta) = \partial_{x_{\gamma_1, \xi_1}} \cdots \partial_{x_{\gamma_M,\xi_M}} \overline{Y}_k(0,\beta).
\]
As shorthand we also write multiplication of the coefficients of $\beta$  for the characteristics of $(\gamma,\xi)$ as 
\[
\beta_{(\gamma,\xi)} = \beta_{\gamma_1, \xi_1} \cdots \beta_{\gamma_M, \xi_M}.
\]
\begin{lemma} \label{l:main1}
\[
\partial_{(\gamma, \xi)} \overline{Y}_k(0,\beta) = \partial_{\gamma} \partial_k V(0) \beta_{(\gamma,\xi)}.
\]
\end{lemma}
\begin{proof}
Lemma~\ref{l:env} establishes
\[
\overline{Y}_k(x, \beta) = \partial_k V(\beta'_1 x_1, \ldots, \beta'_K x_K).
\]
Differentiating with respect to $x_{\gamma_1,\xi_1}$ and evaluating at $x=0$ yields
\[
\partial_{\gamma_1, \xi_1} \overline{Y}_k(0,\beta) = \partial_{\gamma_1} \partial_k V(0) \beta_{\gamma_1, \xi_1}.
\]
By repeating the differentiation process and evaluating at $x=0$ the result follows. Note that in this step, we use the exclusion restriction that the $j$-th regressors $x_j$ are excluded from the desirability indices of the other goods.
\end{proof}

\begin{lemma} \label{l:a2}
\[
\partial_{(\gamma, \xi)} \overline{Y}_k(0) = \partial_{\gamma} \partial_k V(0) \int \beta_{(\gamma, \xi)} d\nu(\beta)
\]
\end{lemma}
\begin{proof}
We obtain
\begin{align*}
    \partial_{(\gamma, \xi)} \overline{Y}_k(0) &=  \partial_{x_{1,\xi_1}} \cdots \partial_{x_{M, \xi_M}}\overline{Y}_k(0) \\
    &= \int \partial_{x_{1,\xi_1}} \cdots \partial_{x_{M, \xi_M}} \overline{Y}_k(0,\beta) d \nu ( \beta) \\
    &= \int \partial_{\gamma} \partial_k V(0) \beta_{(\gamma, \xi)} d \nu ( \beta) \\
    &= \partial_{\gamma} \partial_k V(0) \int \beta_{(\gamma, \xi)} d\nu(\beta)
\end{align*}
where the interchange of integration and differentiation in the second equality follows from Assumption~\ref{a:reg}(i), the third equality is Lemma~\ref{l:main1}, and the final equality follows since the evaluation of $\partial_{\gamma}\partial_k V(0)$ is a constant that does not depend on $\beta$. 
\end{proof}

Combining the result of Lemma~\ref{l:a2} and Assumption~\ref{a:exist} ensures that there exists a set of goods and characteristic indices $(\gamma,\xi)$ such that $\partial_{(\gamma, \xi)} \overline{Y}_k(0) \neq 0$. To see this, recall that Assumption~\ref{a:exist} requires that for each collection of good indices $\gamma$ we can find characteristic indices $\xi$ such that $\int \beta_{(\gamma, \xi)} d\nu(\beta) \neq 0$. Given Assumption~\ref{a:reg}(iv) that $\delta_{\gamma} \delta_k V(0) \neq 0$, Lemma~\ref{l:a2} shows that we must have $\partial_{(\gamma, \xi)} \overline{Y}_k(0) \neq 0$.

\begin{lemma} \label{l:a3}
Fix $j,k \in \{1,\ldots, K\}$ and let $\gamma, \delta \in \{1, K \}^M$. Suppose that each good index shows up exactly the same number of times in $(\gamma,k)$ and $(\delta,j)$. Then
\[
\partial_{\gamma} \partial_k V(0) = \partial_{\delta} \partial_j V(0).
\]
\end{lemma}

\begin{proof}
This is a slight restatement of Lemma~\ref{l:sym}.
\end{proof}
Suppose now that $(\delta, \eta)$ is defined similar to $(\gamma, \xi)$. That is, $\delta \in \{1, \ldots K \}^M$ denotes good indices and $\eta_j \in \{1, \ldots, d_{\delta_{j}}\}$ indexes a characteristic of the $\delta_j$-th good.

Combining the previous two lemmas, we obtain that if each good index shows up exactly the same number of times in $(\gamma,k)$ and $(\delta,j)$, then
\begin{equation} \label{eq:ratio}
\partial_{(\gamma, \xi)} \overline{Y}_k(0) \Bigg/ \partial_{(\delta, \eta)} \overline{Y}_j(0)  =  \int \beta_{(\gamma, \xi)} d\nu(\beta) \Bigg/ \int \beta_{(\delta, \eta)} d\nu(\beta)
\end{equation}
whenever the denominator is nonzero. Thus, if the denominator of moments is identified, the numerator is as well.

\begin{lemma} \label{l:one}
Suppose $\gamma, \delta \in \{1, \ldots, K \}^M$ only differ in at most one component and $\int \beta_{(\delta, \eta)} d\nu ( \beta)$ is identified and nonzero. Then for every $\xi$ tuple of characteristic indices $\int \beta_{(\gamma, \xi)} d\nu (\beta)$ is identified.
\end{lemma}

\begin{proof}
Lemmas~\ref{l:a2} and~\ref{l:a3} immediately imply (\ref{eq:ratio}) since
\begin{align*} 
\partial_{(\gamma, \xi)} \overline{Y}_k(0) \Bigg/ \partial_{(\delta, \eta)} \overline{Y}_j(0)  &= \left( \partial_{\gamma} \partial_k V(0) \int \beta_{(\gamma, \xi)} d\nu(\beta) \right) \Bigg/ \left(\partial_{\delta} \partial_j V(0) \int \beta_{(\delta, \eta)} d\nu(\beta) \right)  \\
&= \int \beta_{(\gamma, \xi)} d\nu(\beta) \Bigg/ \int \beta_{(\delta, \eta)} d\nu(\beta).
\end{align*}
The term $\int \beta_{(\gamma, \xi)} d\nu (\beta)$ is identified because all other parts of (\ref{eq:ratio}) are identified.
\end{proof}

Note that in Lemma~\ref{l:one} that the $\gamma$ and $\delta$ terms can be the same. This covers the non-trivial $K = 1$ case when there are multiple characteristics for the first good.

Theorem~\ref{thm:main} requires that $\int \beta^{M}_{1,1} \nu(d \beta)$ be known. We present a lemma that drops this assumption for the moment. This lemma will be used in subsequent results.
\begin{lemma} \label{l:a5}
If $\int \beta^M_{1,1} d\nu(\beta)$ is not known, then we still identify the ratio of any $M$-th order moments
\[
\int \beta_{(\gamma, \xi)} d\nu(\beta) \Bigg/ \int \beta_{(\delta, \eta)} d\nu(\beta),
\]
provided the denominator is nonzero.
\end{lemma}

\begin{proof}
Start with good indices $\delta^0 = (1, \ldots, 1)$ of length $M$, and characteristic indices $\eta^0$ such that the corresponding moment of $\beta$ is nonzero. Applying Lemma~\ref{l:one} for the pair with goods $\delta^1 = (2, 1, \ldots, 1)$ and characteristic indices $\eta^1$, we identify the ratio
\[
\int \beta_{(\delta^1, \eta^1)} d \nu(\beta) \Bigg/ \int \beta_{(\delta^0, \eta^0)} d\nu(\beta).
\]
We can repeat this procedure with a sequence $(\delta^1,\eta^1)$ and $\delta^2 = (1, 2, \ldots, 1)$ with appropriately chosen characteristic indices $\eta^2$, and so forth, to construct a sequence $\delta^0, \delta^1, \ldots$ that reaches all possible tuples of good indices $\gamma \in \{1, \ldots, M \}^K$. At each step, we can change the good index one component at a time and then apply (\ref{eq:ratio}). This identifies the ratio of two adjacent moments in this sequence. We avoid dividing by zero because of the relevance condition (Assumption~\ref{a:exist}), which implies for each set of goods, $\delta$, we can find tuples of characteristics, $\eta$, such that $\int \beta_{(\delta,\eta)} d\nu(\beta)$ is nonzero.

By multiplication we can identify new ratios. For example, a ratio involving $\delta^2$ and $\delta^0$ is identified via
\begin{align*}
\int \beta_{(\delta^2, \eta^2)} &  d\nu(\beta) \Bigg/ \int \beta_{(\delta^0, \eta^0)} d\nu(\beta) = \\
& \left( \int \beta_{(\delta^2, \eta^2)} d\nu(\beta) \Bigg/ \int \beta_{(\delta^1, \eta^1)} d\nu(\beta)  \right) \left( \int \beta_{(\delta^1, \eta^1)} d\nu(\beta) \Bigg/ \int \beta_{(\delta^0, \eta^0)} d\nu(\beta) \right).
\end{align*}

From these arguments, for each pair of good indices $\gamma$ and $\delta$ we can find some tuples of characteristic indices $\xi$ and $\eta$ such that
\[
\int \beta_{(\gamma, \xi)} d\nu(\beta) \Bigg/ \int \beta_{(\delta, \eta)} d\nu(\beta)
\]
is identified, where numerator and denominator are nonzero.

From Lemma~\ref{l:one}, the ratio
\[
\int \beta_{(\delta, \tilde{\eta})} d\nu(\beta) \Bigg/ \int \beta_{(\delta, \eta)} d\nu(\beta)
\]
is identified for any vector of characteristic indices $\tilde{\eta}$ where $\eta$ is chosen so that the denominator is nonzero. Thus, we identify the ratio of all moments.
\end{proof}
Using Lemma~\ref{l:a5}, we conclude that if we fix $\int \beta^M_{1,1} d\nu (\beta)$ in advance and it is nonzero, we identify all moments. However, we could fix any other nonzero $M$-th order moment and also obtain identification.

Finally, from Lemma~\ref{l:a2} we have for all $\gamma \in \{1, \ldots, K \}^M$ that
\[
\partial_{(\gamma, \xi)} \overline{Y}_k(0) = \partial_{\gamma} \partial_k V(0) \int \beta_{(\gamma, \xi)} d\nu(\beta).
\]
Moreover, from Assumption~\ref{a:exist} we can find some $\xi$ such that the right hand side is nonzero. By dividing, we identify $\partial_{\gamma} \partial_k V(0)$, completing the proof of Theorem~\ref{thm:main}.

\subsection{Proof of Proposition~\ref{prop:alt}} \label{app:propalt}

First, from the envelope theorem (see Lemma~\ref{l:a2} above),
\[
\frac{\partial {\overline{Y}_1(0)}}{\partial x_{1, 1}} = \partial_{1,1} V(0) \int \beta_{1,1} \nu(d \beta).
\]
The function $V$ is convex and hence $\partial_{1,1} V(0) > 0$ whenever this derivative is nonzero, so $\frac{\partial {\overline{Y}_1(0)}}{\partial x_{1, 1}}$ and $\int \beta_{1,1} d\nu(\beta)$ have the same sign. Thus, the sign of the first moment of $\beta_{1,1}$ is identified from above and the magnitude is assumed known in Assumption~\ref{a:oneind}.

We prove the remainder of the result by induction on $M$. Recall that with $M = 1$, first order moments are identified from Theorem~\ref{thm:main} using the assumption that $\int \beta_{1,1} \nu (d \beta)$ is known and nonzero.

Now, fix an $M$ such that $1 \le M \leq \overline{M}-1$. As the inductive hypothesis, we assume all $M$-th order moments $\int \beta_{(\delta, \eta)} d\nu (\beta)$ are identified for all $\delta \in \{1,\ldots,K\}^M$ and $\eta$ collections of characteristic indices. We show all $M+1$ order moments are also identified. 

By Assumption~\ref{a:exist}, when $K \geq 2$, for $\delta \in \{ 2, \ldots, K \}^M$ (i.e. no good index is equal to $1$) we can find a collection of characteristic indices $\eta$ such that $\int \beta_{(\delta, \eta)} d\nu (\beta) \neq 0$. If instead $K = 1$, we can set $\delta$ as the length-$M$ vector of $1$'s and let $\eta$ be some collection of characteristic indices with $\eta_{m} \neq 1$ such that
\[
\int \beta_{1,\eta_{1}} \cdots \beta_{1, \eta_{M}} d\nu(\beta) \neq 0.
\]
In either case $K = 1$ or $K \geq 2$, set $\tilde{\delta} = (\delta', 1)'$ and $\tilde{\eta} = (\eta',1)'$. Then we obtain
\[
\int \beta_{\left(\tilde{\delta}, \tilde{\eta} \right)} d\nu(\beta) =  \int \beta_{1,1} d\nu(\beta) \int \beta_{(\delta, \eta)} d\nu(\beta)
\]
because $\beta_{1,1}$ is independent of all other components of $\beta$ under the measure $\nu$, and the tuple $(\delta,\eta)$ does not include the first characteristic of good $1$. In particular, we identify
\[
\int \beta_{(\tilde{\delta}, \tilde{\eta})} d\nu(\beta),
\]
which is nonzero because it is the product of two nonzero terms. From Lemma~\ref{l:a5}, we identify the ratio of all $M+1$ order moments to $\int \beta_{(\tilde{\delta}, \tilde{\eta})} d\nu(\beta)$. Since $\int\beta_{(\tilde{\delta}, \tilde{\eta})} d\nu(\beta)$ is known and nonzero we identify all $M + 1$ order moments from Theorem~\ref{thm:main}. This establishes identification of the moments. To identify derivatives of $V$, use Lemma~\ref{l:a2} as in the proof of Theorem~\ref{thm:main}.

\subsection{Proof of Proposition~\ref{prop:idhomog}}
Let $x = (x_1', \ldots, x_K')'$ satisfy $x_{k,1} \in [\underline{x}_{k,1}, \overline{x}_{k,1}]$ for each $k$, and $x_{k,j} = 0$ for $j > 1$. Let $x_{:,1}=(x_{1,1},\ldots,x_{K,1})$ be a vector of the first characteristics for each good. From Lemma~\ref{l:env} and integrating over $\beta$, we obtain
\[ \overline{Y}(x) = \nabla V(x_{:,1})\]
where $V(x)$ is convex. 

Consider initial characteristic values, $x^I$, and final characteristic values, $x^F$, such that for all $k\in \{1,\ldots, K\}$ and for all $j>1$, the equality $x_{k,j}^I=x_{k,j}^F=0$ holds. By integrating from $x^I$ to $x^F$, we obtain
\begin{align*} V(x_{:,1}^F)-V(x_{:,1}^I)&=
\int_{0}^1 \overline{Y}(tx^F-(1-t)x^I) \cdot (x_{:,1}^F-x_{:,1}^I) dt,
\end{align*}
where Riemann integrability follows from \cite{rockafellar1970convex} Corollary 24.2.1.

\section{Supplemental Results}

\subsection{Plug-in Estimation} \label{app:estimation}
The proof of Theorem~\ref{thm:main} is based on multiplying derivative ratios. By directly pluging-in an estimator of the $M$-th order derivatives, one can construct an estimator of the $M$-th order moments. Suppose we have estimator
\[
\partial_{x_{k_1,\ell_1}} \cdots \partial_{x_{k_M,\ell_M}} \hat{\overline{Y}}_1(0)
\]
of the associated $M$-th order derivative of the average structural function where $k_m=1$ for all but one term which has $k_{\tilde{m}}=j$. In addition, suppose we have an estimator
\[
\partial_{x_{1,1}} \cdots \partial_{x_{1,1}} \hat{\overline{Y}}_j(0)
\]
of the $M$-th order partial derivative of the structural function with respect to the $x_{1,1}$ regressor.

We construct an estimator
\[
\widehat{\int \beta_{k_1,\ell_1} \cdots \beta_{k_M,\ell_M} \nu (d \beta)} = \frac{\partial_{x_{k_1,\ell_1}} \cdots \partial_{x_{k_M,\ell_M}} \hat{\overline{Y}}_1(0)}{\partial_{x_{1,1}} \cdots \partial _{x_{1,1}} \hat{\overline{Y}}_j(0)},
\]
where for simplicitly we assume $\int \beta^M_{1,1} \nu (d \beta) = 1$. (More generally we need it to be known \textit{a priori} and nonzero.)

Using Equation~\ref{eq:ratio}, we see that 
\begin{align*}
\int \beta_{k_1,\ell_1} \cdots & \beta_{k_M,\ell_M} \nu (d \beta) - \widehat{\int \beta_{k_1,\ell_1} \cdots \beta_{k_M,\ell_M} \nu (d \beta)} = \\
& \frac{\partial_{x_{k_1,\ell_1}} \cdots \partial_{x_{k_M,\ell_M}} \overline{Y}_1(0)}{\partial_{x_{1,1}} \cdots \partial _{x_{1,1}} \overline{Y}_j(0)} - \frac{\partial_{x_{k_1,\ell_1}} \cdots \partial_{x_{k_M,\ell_M}} \hat{\overline{Y}}_1(0)}{\partial_{x_{1,1}} \cdots \partial _{x_{1,1}} \hat{\overline{Y}}_j(0)}.
\end{align*}
Thus, estimation error on the right-hand side translates to estimation error on the left hand side for the $M$-th order moment of $\beta$. Note that the choice of the $j$-th good is arbitrary here, and so one could also construct an estimator with right hand side replaced by an average over ratios with respect to different goods.

This argument can be generalized to additional moments. Here, we use the fact that we are interested in an $M$-th order moment that is only one good index away from being a vector of $1$'s. For other $M$-th order moments, our constructive formulas show one must multiply additional derivative ratios. See in particular the proof of Lemma~\ref{l:a5}.

\subsection{$V$ Known and Relation to \cite{fox2012random}} \label{s:vknown}
Theorem~\ref{thm:main} identifies the $M$-th order moment of $\beta$ when we fix $\int \beta^M_{1,1} \nu (d \beta)$. By fixing the entire distribution of $\beta_{1,1}$, we can identify all moments of $\beta$. An alternative assumption is that $V$ is known. If we impose this assumption, then we can drop Assumption~\ref{a:scale}, which provides knowledge of each $M$-th order moment of $\beta_{1,1}$, and Assumption~\ref{a:exist} that a rich collection of moments are nonzero. The intuition why we can relax the scale assumption on moments is that here we instead set the scale of $V$.
\begin{assm} \label{a:vknown}
$V$ is known in a neighborhood of $0$, up to an additive constant.
\end{assm}
\begin{prop} \label{p:vknown}
Let Assumptions~\ref{a:factor}-\ref{a:cont},~\ref{a:reg}, and~\ref{a:vknown} hold with the same natural number $M$. Each $M$-th order moment
\[
\int \beta_{k_1,\ell_1} \cdots \beta_{k_M, \ell_M} d\nu ( \beta)
\]
is identified.
\end{prop}
\begin{proof}
Lemma~\ref{l:a2} holds under the assumptions of this proposition, and so for each tuple of good indices $\gamma \in \{1, \ldots, M\}^M$ and characteristic indices $\xi$ we have
\[
\partial_{(\gamma, \xi)} \overline{Y}_k(0) = \partial_{\gamma} \partial_k V(0) \int \beta_{(\gamma, \xi)} d\nu( \beta).
\]
Since $\partial_{\gamma} \partial_k V(0) \neq 0$ by assumption, we identify $\int \beta_{(\gamma, \xi)} \nu(d \beta)$ by
\[
\partial_{(\gamma, \xi)} \overline{Y}_k(0) \Bigg/ \partial_{\gamma} \partial_k V(0) = \int \beta_{(\gamma, \xi)} d\nu(\beta).
\]
\end{proof}
The proof demonstrates that in fact, we only need $\partial_{\gamma} \partial_k V(0)$ to be known and nonzero for \textit{some} $k$ in order to identify a corresponding moment of $\beta$. For further relation to Theorem~\ref{thm:main}, consider good indices $\gamma = (1, \ldots, 1)$ and characteristic indices $\xi = (1, \ldots, 1)$. Then as in the proof of Proposition~\ref{p:vknown}, we obtain
\[
\partial^M_{1} \overline{Y}_k(0) = \partial^M_{1} \partial_k V(0) \int \beta^M_{1,1} d\nu(\beta).
\]
This shows that one can \textit{either} fix the $M$-th moment of $\beta_1$ or fix $\partial^M_{1} \partial_k V(0)$ for some $k$, and then the other can be identified. Thus, Assumption~\ref{a:scale} can be replaced in Theorem~\ref{thm:main} if we instead assume $\partial^M_{1} \partial_k V(0)$ is known for some $k$. Alternatively, if we assume $\beta_{1,1}$ is independent of $\beta$ as in Section~\ref{sec:alt}, then specifying $\partial_{1,1} V(0)$ identifies $\int \beta_{1,1} d\nu (\beta)$ by the envelope theorem
\[
\frac{\partial {\overline{Y}_1(0)}}{\partial x_{1, 1}} = \partial_{1,1} V(0) \int \beta_{1,1} d\nu(\beta).
\]
Thus, independence combined with a single scale assumption on a partial derivative of $V$ can identify all moments of $\beta$ by adapting Proposition~\ref{prop:alt}.

In discrete choice, \cite{fox2012random} present a constructive approach to identifying moments of the distribution of random coefficients. Specializing our analysis to discrete choice, their assumptions show that when the distribution of an additive error is known (e.g. logit with known intercept) that this implies identification of $V$. To see this, recall $V$ is defined as an indirect utility function given a disturbance $\overline{D}$ and constraint set $\overline{B}$. In turn, using Lemma~\ref{l:agg} we see $\overline{D}$ and $\overline{B}$ are determined by the budget set $B$, disturbance function $D$, and measure $\mu$ over $\eps$. Thus, when the budget set and $\mu$ are known then one can find $\overline{B}$ and $\overline{D}$ needed to compute $V$. For example,  multinomial logit is described by 
\[
V(\vec{u}) = \max_{y \in \overline{B}} \sum_{k = 1}^K y_k u_k + \sum_{k = 1}^K(\alpha_k + p_k \ln p_k),
\]
where $\alpha_k$ is a nonrandom intercept for good $k$ and $\overline{B}$ is the probability simplex (e.g. \cite{anderson1992discrete}). The derivatives of $V$ can be used to yield the standard logit formula
\[
\overline{Y}_k(\vec{u}) = \frac{e^{\alpha_k + u_k}}{\sum_{j = 1}^K e^{\alpha_j + u_j}}.
\]

We conclude that Proposition~\ref{p:vknown} is a generalization of a technique of \cite{fox2012random} to settings outside of discrete choice.\footnote{\cite{fox2012random} also presents nonconstructive results using alternative assumptions maintaining the assumption that $V$ is known.} However, Theorem~\ref{thm:main} does \textit{not} require that $V$ be known. Thus, the results in this paper complement their approach since while it relaxes assumptions on $V$, it  instead requires an additional scale assumption (Assumption~\ref{a:scale}) and requires that a rich collection of moments are nonzero (Assumption~\ref{a:exist}).

\subsection{Homogeneity of Coefficients and Relation to \cite{chernozhukov2019nonseparable}} \label{supp:cherno}
When $M = 2$, the proof of Theorem~\ref{thm:main} establishes the constructive formula
\begin{align} \label{eq:ratioid}
\frac{\partial {\overline{Y}_k(0)}}{\partial x_{j,\ell}} \Bigg/ \frac{\partial {\overline{Y}_{j}(0)}}{\partial x_{k,m}}  = \int \beta_{j,\ell} d\nu( \beta) \Bigg/ \int \beta_{k,m} d\nu( \beta).
\end{align}
A version of (\ref{eq:ratioid}) has appeared for binary choice in \cite{chernozhukov2019nonseparable}, who also discuss identification of the ratios of $M$-th order moments of $\beta$ up to scale. They also mention one can identify certain moments up to scale in multinomial choice. Here is one interpretation of their discussion, translated to our setup. Start with a second order derivative like
\[
\partial_{x_{k_1,\ell_1}} \partial_{x_{k_2,\ell_2}} \overline{Y}_{k_3}(0) = \partial_{k_1} \partial_{k_2} \partial_{k_3} V(0)  \int \beta_{k_1, \ell_1} \beta_{k_2, \ell_2} d\nu (\beta).
\]
Now keep the good indices ($k_1$ and $k_2$) constant, but change the characteristics to get
\[
\partial_{x_{k_1,\tilde{\ell}_1}} \partial_{x_{k_2,\tilde{\ell}_2}} \overline{Y}_{k_3}(0) = \partial_{k_1} \partial_{k_2} \partial_{k_3} V(0)  \int \beta_{k_1, \tilde{\ell}_1} \beta_{k_2, \tilde{\ell}_2} d\nu (\beta).
\]
Since the derivatives of $V$ are taken with respect to the same arguments, we can divide these equations to identify the associated ratios of moments of $\beta$. This technique resembles an implicit function theorem argument for identification. Importantly, this technique only covers ratios of moments in which the good indices ($k_1$ and $k_2$ here) are the same, because it does not use symmetry (cf. Lemma~\ref{l:sym}). Using symmetry, this paper establishes identification of the ratio of \textit{all} $M$-th order moments, not only those that have the same good indices. However, if we impose additional assumptions such as $\beta_j = \beta_k$ for all goods, then the choice of good indices does not matter. In this special case, using the equations described previously one can identify the ratio of any $2$-nd order moments of $\beta$. Similar arguments can identify the ratio of any $M$-th order moments of $\beta$.

\end{appendices}

\end{document}